\documentclass[11pt,english]{article}
\usepackage[T1]{fontenc}
\usepackage[latin9]{inputenc}
\usepackage{geometry}
\usepackage{amsmath}
\usepackage{amsthm}
\usepackage{amssymb}
\usepackage{esint}

\makeatletter

\newcommand{\noun}[1]{\textsc{#1}}

\theoremstyle{plain}
\newtheorem{thm}{\protect\theoremname}[section]
  \theoremstyle{definition}
  \newtheorem{defn}[thm]{\protect\definitionname}
  \theoremstyle{plain}
  \newtheorem{cor}[thm]{\protect\corollaryname}
  \theoremstyle{plain}
  \newtheorem{conjecture}[thm]{\protect\conjecturename}
  \theoremstyle{remark}
  \newtheorem{rem}[thm]{\protect\remarkname}

\makeatother

\usepackage{babel}
  \providecommand{\conjecturename}{Conjecture}
  \providecommand{\corollaryname}{Corollary}
  \providecommand{\definitionname}{Definition}
  \providecommand{\remarkname}{Remark}
\providecommand{\theoremname}{Theorem}

\begin{document}

\title{Spectral and Modular Analysis of \#P Problems}

\author{Ohad Asor}
\maketitle
\begin{abstract}
We present various analytic and number theoretic results concerning
the \noun{\#sat} problem as reflected when reduced into a \noun{\#part}
problem. As an application we propose a heuristic to probabilistically
estimate the solution of \noun{\#sat} problems.
\end{abstract}

\section{Overview}

\noun{\#sat} is the problem of counting the number of satisfying assignments
to a given 3CNF formula, while \noun{\#part} is the problem of counting
the number of zero partitions in a given set of integers. Precise
definitions will be given later on. We present various results concerning
\noun{\#part} and analyze their connection with \noun{\#sat}. On section
2 we skim some preliminaries. Section 3 presents the core of the analytic
setting by analyzig the \noun{\#part} problem as manipulations over
product of cosines. Section 4 derives a modular-arithmetic formula
for computing \noun{\#part}, and section 5 presents implications to
complexity theory. Section 6 deals with asymptotic normality, and
section 7 deals with variances and correlations. Section 8 propose
how multiple reductions may give probabilistic answer to \noun{\#sat}
as a consequence our analysis. Section 9 summarizes the highlights
of the paper.

\section{Preliminaries}

Our setting is counting the number of solutions given an instance
of the Partition problem. We sometimes use custom terminology as there
is no unified one.
\begin{defn}
Given $n\in\mathbb{N},\mathbf{x}\in\mathbb{N}^{n}$, a \emph{Partition}
$\sigma$ of $\mathbf{x}$ is some $\sigma\in\left\{ -1,1\right\} ^{n}$.
The \emph{size} of the partition $\sigma$ is $\left\langle \mathbf{x},\sigma\right\rangle =\sum_{k=1}^{n}\sigma_{k}x_{k}$.
A partition is called a \emph{zero partition} if its size is zero.
The problem \noun{\#part} is to determine the number of zero partitions
given $\mathbf{x}$. The problem \noun{part} is deciding whether a
zero partition exists or not for $\mathbf{x}$. The \emph{Weak} setting
of the problem is when $\mathbf{x}$ is supplied in unary radix, and
the \emph{Strong} setting is when it is supplied in binary radix (or
another format with same efficiency), therefore the input size is
logarithmically smaller on the strong setting.
\end{defn}
\noun{\#part} is in \#P complexity class. The setting of \noun{\#part}
after being reduced from the counting Boolean Satisfiability problem
(\noun{\#sat}) is $n$ integers to partition each having up to ${\cal O}\left(n\right)$
binary digits (where $n$ is linear in the size of the CNF formula),
demonstrating why the rather strong setting is of interest. In fact,
there exist polynomial time algorithms solving the weak setting of
\noun{part}, notably Dynamic Programming algorithms, as well as the
formula derived on Theorem 4.1 below. However, solving \noun{part}
under the strong setting is not possible in polynomial time (as a
function of the input's length), unless P=NP.

\noun{\#sat} can be reduced to \noun{\#subset-sum} using an algorithm
described in \cite{key-1}, while various slight variations appear
on the literature. We summarize the reductions on the Appendix.

\section{Analytic Setting}
\begin{thm}
Given $n\in\mathbb{N},\mathbf{x}\in\mathbb{N}^{n}$ then the (probability-theoretic)
characteristic function of the random variable $\left\langle \mathbf{x},\sigma\right\rangle =\sum_{k=1}^{n}x_{k}\sigma_{k}$
over uniform $\sigma\in\left\{ -1,1\right\} ^{n}$ is $\prod_{k=1}^{n}\cos\left(x_{k}t\right)$.\end{thm}
\begin{proof}
Consider the formula $2\cos a\cos b=\cos\left(a+b\right)+\cos\left(a-b\right)$
and the cosine being even function to see that: 
\begin{equation}
\psi\left(t\right)\equiv\psi\left(x_{1},\dots,x_{n},t\right)\equiv\prod_{k=1}^{n}\cos\left(x_{k}t\right)=2^{-n}\sum_{\sigma\in\left\{ -1,1\right\} ^{n}}\cos\left(t\left\langle \mathbf{x},\sigma\right\rangle \right)=\mathbb{E}\left[e^{it\left\langle \mathbf{x},\sigma\right\rangle }\right]
\end{equation}
\end{proof}
\begin{cor}
Given $n\in\mathbb{N},\mathbf{x}\in\mathbb{N}^{n}$ then
\begin{equation}
2^{n}\intop_{0}^{1}\prod_{k=1}^{n}\cos\left(2\pi x_{k}t\right)dt
\end{equation}

is the number of zero partitions of $\mathbf{x}$.\end{cor}
\begin{proof}
Following (1) and integrating both sides. Stronger statements are
possible (e.g. for $\mathbf{x}\in\mathbb{R}^{n}$ or $\mathbf{x}\in\mathbb{C}^{n}$)
using characteristic function inversion theorems.\end{proof}
\begin{thm}
Given $\left\{ n,N\right\} \subset\mathbb{N},j\in\mathbb{Z},\mathbf{x}\in\mathbb{N}^{n}$
then
\begin{equation}
\frac{2^{n}}{N}\sum_{m=1}^{N}e^{2\pi ij\frac{m}{N}}\prod_{k=1}^{n}\cos\left(2\pi x_{k}\frac{m}{N}\right)
\end{equation}

is the number of partitions of $\mathbf{x}$ having size that is divisible
by $N$ with remainder $j$.\end{thm}
\begin{proof}
Following (1) 
\begin{equation}
\frac{2^{n}}{N}\sum_{m=1}^{N}\prod_{k=1}^{n}\cos\left(2\pi x_{k}\frac{m}{N}\right)=\sum_{\sigma\in\left\{ -1,1\right\} ^{n}}\frac{1}{N}\sum_{m=1}^{N}e^{2\pi i\frac{m}{N}\left\langle \mathbf{x},\sigma\right\rangle }
\end{equation}

The sum of the roots of unity on the rhs of (4) is zero if $N$ does
not divide $\left\langle \mathbf{x},\sigma\right\rangle $, and is
one if $N$ does divide it, therefore (4) is equal to
\begin{equation}
\sum_{u=-\infty}^{\infty}c_{uN}
\end{equation}
where $c_{u}$ denotes the number of partitions that sum to $u$:
\begin{equation}
c_{u}=\left|\left\{ \sigma\in\left\{ -1,1\right\} ^{n}|\left\langle \mathbf{x},\sigma\right\rangle =u\right\} \right|
\end{equation}
As for the remainder, observe that
\begin{equation}
\sum_{\sigma\in\left\{ -1,1\right\} ^{n}}e^{2\pi it\left(\left\langle \mathbf{x},\sigma\right\rangle +j\right)}=e^{2\pi itj}2^{n}\prod_{k=1}^{n}\cos\left(2\pi x_{k}t\right)
\end{equation}
\end{proof}
\begin{conjecture}
For all even $n$, for all $\mathbf{x}\in\mathbb{N}^{n}$ the number
of $\mathbf{x}$'s zero partitions is no more than the number of zero
partitions of vector of size $n$ with all its elements equal $1$.
Namely, never more than $\binom{n}{\frac{1}{2}n}$ zero partitions.

Furthermore, for all odd $n$, for all $\mathbf{x}\in\mathbb{N}^{n}$
the number of $\mathbf{x}$'s zero partitions is no more than the
number of zero partitions of vector of size $n$ with all its elements
equal $1$ except one element that equals $2$.
\end{conjecture}

\section{Modular Arithmetic Formula}
\begin{thm}
Given $n\in\mathbb{N},\mathbf{x}\in\mathbb{N}^{n}$, the number of
$\mathbf{x}$'s zero partition out of all possible $2^{n}$ partitions
is encoded as a binary number in the binary digits of 
\begin{equation}
\prod_{k=1}^{n}\left[1+4^{nx_{k}}\right]
\end{equation}
 from the $s$'th dight to the $s+n$ digit, where $s=n\left\langle \mathbf{x}\mathbf{,}1\right\rangle $.\end{thm}
\begin{proof}
We write down the following sum and perform substitution according
to (1):
\begin{equation}
S=\frac{1}{n}\sum_{m=1}^{n}\psi\left(\frac{2\pi m}{n}+i\ln2\right)=\sum_{\sigma\in\left\{ -1,1\right\} ^{n}}\frac{2^{-\left\langle \mathbf{x},\sigma\right\rangle }}{n}\sum_{m=1}^{n}e^{\frac{2\pi im}{n}\left\langle \mathbf{x},\sigma\right\rangle }
\end{equation}

multiplying all $x_{k}$ by $n$ (while preserving partitions, since
we can always multiply all $x_{k}$ by the same factor and keep the
exact number of zero partitions) puts $e^{\frac{2\pi im}{n}\left\langle n\mathbf{x},\sigma\right\rangle }=1$
and we get:
\begin{equation}
S=\sum_{\sigma\in\left\{ -1,1\right\} ^{n}}2^{-n\left\langle \mathbf{x},\sigma\right\rangle }=\sum_{u=-\infty}^{\infty}c_{u}2^{-u}
\end{equation}

where $c_{u}$ is defined in. Recalling that $\sum_{u=-\infty}^{\infty}c_{u}=2^{n}$
and $c_{u}$ are all positive, while on (10) being multiplied by distinct
powers $2^{\pm n}$, therefore the summands' binary digits never interfere
with each other. Recalling that $c_{0}$ is our quantity of interest,
we have shown that the number of zero partitions in $\mathbf{x}$
is encoded at
\begin{equation}
\left\lfloor \frac{2^{n}}{n}\sum_{m=1}^{n}\prod_{k=1}^{n}\cos\left[nx_{k}\left(\frac{2\pi m}{n}+i\ln2\right)\right]\right\rfloor \mod2^{n}
\end{equation}
\begin{equation}
=\left\lfloor \prod_{k=1}^{n}\left[2^{nx_{k}}+2^{-nx_{k}}\right]\right\rfloor \mod2^{n}
\end{equation}
\begin{equation}
=\left\lfloor 2^{-n\sum_{k=1}^{n}x_{k}}\prod_{k=1}^{n}\left[1+2^{2nx_{k}}\right]\right\rfloor \mod2^{n}
\end{equation}
Set
\begin{equation}
M=\prod_{k=1}^{n}\left[1+2^{2nx_{k}}\right]=\sum_{\sigma\in\left\{ 0,1\right\} ^{n}}2^{2n\left\langle \mathbf{x},\sigma\right\rangle }
\end{equation}
then (12) tells us that the number of zero partitions is encoded as
a binary number in the binary digits of $M$, from the $s$'th dight
to the $s+n$ digit.
\end{proof}
Note that the substitution in (9) could take a simpler form. Put $t=i\ln2$
in (1):
\begin{equation}
\prod_{k=1}^{n}\cosh\left(x_{k}\ln2\right)=\prod_{k=1}^{n}\left[2^{x_{k}}+2^{-x_{k}}\right]=\mathbb{E}\left[2^{\left\langle \mathbf{x},\sigma\right\rangle }\right]
\end{equation}

\section{Hardness of Integration}
\begin{cor}
$\mathbf{x}\in\mathbb{Q}^{n}$ has a zero partition if and only if
\begin{equation}
\intop_{0}^{\infty}\prod_{k=1}^{n}\cos\left(x_{k}t\right)dt=\infty
\end{equation}
and does not have a zero partition if and only if
\begin{equation}
\intop_{0}^{\infty}\prod_{k=1}^{n}\cos\left(x_{k}t\right)dt=0
\end{equation}
\end{cor}
\begin{proof}
Follows from Theorem 4.1, the integrand being periodic, change of
variable to support rationals, and the integral over a single period
being nonnegative for all inputs.\end{proof}
\begin{cor}
There is no algorithm that takes any function that can be evaluated
in polynomial time, and decides in polynomial time whether its integral
over the real line is zero (conversley, infinity) unless P=NP.\end{cor}
\begin{thm}
{[}Theorem 2.2 on \cite{key-3}{]} If $u$ is an analytic function
satisfying $\left|u\left(z\right)\right|\leq M$ in $\frac{1}{r}\leq\left|z\right|\leq r$
($z\in\mathbb{C}$) for some $r>1$, then for any $N\geq1$ the trapezoid
rule with $N$ points will be far from the exact integral by no more
than $\frac{4\pi M}{r^{N}-1}$.\end{thm}
\begin{cor}
If for every \noun{\#part} instance it is possible to efficiently
find a function $w$ such that given $\psi$ as in (1) that corresponds
the problem's instance, and 
\begin{equation}
u\left(z\right)=\psi\left(w\left(z\right)\right)w'\left(z\right)
\end{equation}
is computable in polynomial time (wrt the input length and the desired
output accuracy) and satisfies the conditions of Theorem 4.4 with
$r=2$ and $M={\cal O}\left(\text{poly}\left(\sum_{k=1}^{n}x_{k}\right)\right)$,
then P=NP.\end{cor}
\begin{proof}
Observe that $\psi$ behaves like $e^{x_{k}t}$ for imaginary input.
It therefore satisfies $M=e^{r\sum_{k=1}^{n}x_{k}}$ at the setting
of Theorem 4.4. For exponential convergence wrt \noun{part}'s input
length we need $\frac{4\pi M}{r^{N}-1}$ diminish exponentially. Therefore
if we can change the variable of integration in (17) using some $w$
and result with $M={\cal O}\left(\text{poly}\left(\sum_{k=1}^{n}x_{k}\right)\right)$,
we could estimate the integral in (17) to our desired accuracy ($2^{-n}$)
in subexponential time. \end{proof}
\begin{rem}
The desired accuracy mentioned in Corollary 5.4 is the same accuracy
desired from the integral (typically $n$ binary digits for our integrand,
as (1) suggests). This is due to Kahan summation algorithm (\cite{key-9}).
We can compute the integrand only up to that accuracy when we use
the trapezoid rule, as long as we perform the summation according
to Kahan's algorithm (in constant multiplicative cost). It is evident
that $\psi$ as for itself can be computed in polynomial time to the
desired accuracy for all real $t\in\left(0,1\right)$, even under
the strong setting of \noun{\#part}:

Denote by $M\left(n\right)$ the complexity of multiplying two $n$-digit
numbers up to accuracy of $2n$. Then multiplying three numbers can
be done by multiplying the first two in no more than $M\left(n\right)$
and taking only the first $n$ digits of the result. Afterwards we're
left again with two $n$-digit numbers to multiply, ending with total
complexity of no more than $2M\left(n\right)$. Continuing this way,
the complexity of multiplying $n$ numbers up to precision of $n$
digits takes no more than ${\cal O}\left(\left(n-1\right)M\left(n\right)\right)$.
Note that the multiplicands need not be more accurate than $n$ digits,
since higher digits won't impact lower digits in the result as long
as we multiply numbers in $\left(0,1\right)$, as in cosine. As for
computing every single cosine, observe that $\cos2^{-n}=\sum_{k=0}^{\infty}\frac{1}{\left(2k\right)!}4^{-nk}$
prescribes the digits of the result nicely right away up to a single
division and with linearly growing precision. It can also be achieved
directly from the input's digits, by writing $x=\sum_{k}d_{k}2^{-k}\implies\prod_{k}\left(e^{2^{-k}}\right)^{d_{k}}$
where in binary we have $d_{k}\in\left\{ 0,1\right\} $, suggesting
$\cos2^{\pm n}$ to be precomputed. We also note that the formulas
for $\cos\left(a+b\right),\sin\left(a+b\right)$ can be applied to
calculate the trigonometric functions of $n$-digit binary number
in linear amount of arithmetic operations, by simply following its
$1$ digits and taking \textbf{$b=2^{-k}$} for all $k$ up to $n$.
Therefore computing the cosine in concern is ${\cal O}\left(nM\left(n\right)\right)$
per one input, so we end up with complexity of maximum ${\cal O}\left(n^{2}M^{2}\left(n\right)\right)\approx{\cal O}\left(n^{5}\right)$
per computing the integrand once up to the desired accuracy. Recalling
that \#SAT grows quadratically when reduced to \#PART, we reach ${\cal O}\left(\ell^{10}\right)$
per single integrand evaluation where $\ell$ is the length of the
CNF formula.
\end{rem}

\section{Asymptotic Normality }

Observe that the Conjecture 3.4 says that for all $\mathbf{x}\in\mathbb{N}^{n}$
we have
\begin{equation}
\intop_{0}^{1}\prod_{k=1}^{n}\cos\left(2\pi x_{k}t\right)dt\leq\intop_{0}^{1}\cos^{n}\left(2\pi t\right)dt
\end{equation}
note that $\intop_{0}^{1}\cos^{n}\left(t\right)dt$ approaches to
a gaussian as $n$ tends to infinity:

\begin{equation}
\lim_{n\rightarrow\infty}\intop_{0}^{\sqrt{n}}\cos^{n}\frac{2\pi t}{\sqrt{N}}dt=\lim_{n\rightarrow\infty}\intop_{0}^{\sqrt{n}}\left[1-\frac{4\pi^{2}t^{2}}{2n}+{\cal O}\left(\frac{1}{n}\right)\right]^{n}dt
\end{equation}
\begin{equation}
=\intop_{0}^{\infty}e^{-2\pi^{2}t^{2}}dt=\frac{1}{\sqrt{8\pi}}\approx0.1994
\end{equation}
as the standard Fourier transform derivation of the Central Limit
Theorem suggests. Similarly, if we take a vector $\mathbf{x}$ and
equally add more copies of its elements up to infinity (e.g. transforming
$\left\{ 1,2,3\right\} $ into $\left\{ 1,1,1,2,2,2,3,3,3\right\} $),
we get asymptotic amount of zero partition written as:

\begin{equation}
\lim_{N\rightarrow\infty}\intop_{0}^{\sqrt{N}}\prod_{k=1}^{n}\cos^{N}\left(\frac{2\pi x_{k}t}{\sqrt{N}}\right)dt
\end{equation}
note that now the limit is wrt $N$ since we still have base $n$
numbers, just copied $N$ times. Continuing:
\begin{equation}
=\lim_{N\rightarrow\infty}\intop_{0}^{\sqrt{N}}\prod_{k=1}^{n}\left[1-\frac{4\pi^{2}x_{k}^{2}t^{2}}{2N}+{\cal O}\left(\frac{1}{N}\right)\right]^{N}dt
\end{equation}
\begin{equation}
=\intop_{0}^{\infty}e^{-2\pi^{2}t^{2}\sum_{k=1}^{n}x_{k}^{2}}dt=\frac{1}{\sqrt{8\pi\sum_{k=1}^{n}x_{k}^{2}}}
\end{equation}
It is interesting to see that the resulted gaussian is diagonalized,
i.e. no correlations between the $x_{k}$'s at the asymptote on this
special case of having infinitely many copies. This means that the
fact that numbers are being copied will always govern any other property
of the numbers, except the single quantity $\sqrt{\sum_{k=1}^{n}x_{k}^{2}}$.

\section{Second Order Statistics}
\begin{thm}
Given $n\in\mathbb{N},N\in\mathbb{N},\mathbf{x}\in\mathbb{N}^{n}$,
the variance of the sizes of all partitions is the sum of the squares
of the input. Formally:
\begin{equation}
\sum_{k=1}^{n}x_{k}^{2}=2^{-n}\sum_{\sigma\in\left\{ -1,1\right\} ^{n}}\left\langle \mathbf{x},\sigma\right\rangle ^{2}
\end{equation}
while
\begin{equation}
\frac{2^{n}}{N^{3}}\sum_{m=1}^{N}\left.\frac{\partial^{2}}{\partial t^{2}}\prod_{k=1}^{n}\cos\left(2\pi x_{k}t\right)\right|_{t=\frac{m}{N}}
\end{equation}
is the variance of the sizes of all partitions that their size is
divisible by $N$ without remainder.\end{thm}
\begin{proof}
Following (1) and differentiating:
\begin{equation}
\prod_{k=1}^{n}\cos\left(\pi x_{k}t\right)=2^{-n}\sum_{\sigma\in\left\{ -1,1\right\} ^{n}}\cos\left(\pi t\left\langle \mathbf{x},\sigma\right\rangle \right)
\end{equation}
\begin{equation}
\implies\sum_{\ell=1}^{n}x_{\ell}\sin\left(\pi x_{\ell}t\right)\prod_{k\neq\ell}^{n}\cos\left(\pi x_{k}t\right)=2^{-n}\sum_{\sigma\in\left\{ -1,1\right\} ^{n}}\left\langle \mathbf{x},\sigma\right\rangle \sin\left(\pi t\left\langle \mathbf{x},\sigma\right\rangle \right)
\end{equation}
\begin{equation}
\implies\sum_{\ell=1}^{n}\sum_{\ell\prime=1}^{n}-x_{\ell}\sin\left(\pi x_{\ell}t\right)x_{\ell\prime}\sin\left(\pi x_{\ell\prime}t\right)\prod_{k\neq\ell,\ell\prime}^{n}\cos\left(\pi x_{k}t\right)+x_{\ell}^{2}\prod_{k=1}^{n}\cos\left(\pi x_{k}t\right)
\end{equation}
\begin{equation}
=2^{-n}\sum_{\sigma\in\left\{ -1,1\right\} ^{n}}\left\langle \mathbf{x},\sigma\right\rangle ^{2}\cos\left(\pi t\left\langle \mathbf{x},\sigma\right\rangle \right)
\end{equation}

and (19) follows by substituting $t=0$. (19) can be proved using
Parseval identity as well. Turning to (20):
\begin{equation}
\frac{2^{n}}{N}\sum_{m=1}^{N}\left.\frac{\partial^{2}}{\partial t^{2}}\prod_{k=1}^{n}\cos\left(2\pi x_{k}t\right)\right|_{t=\frac{m}{N}}=\sum_{\sigma\in\left\{ -1,1\right\} ^{n}}\left\langle \mathbf{x},\sigma\right\rangle ^{2}\cos\left(2\pi\frac{m}{N}\left\langle \mathbf{x},\sigma\right\rangle \right)=\sum_{u=-\infty}^{\infty}u^{2}N^{2}c_{Nu}
\end{equation}
due to aliasing of roots of unity, and $c_{Nu}$ the number of partitions
whose size is divisible by $Nu$ as in (4).\end{proof}
\begin{rem}
It is easy to derive all moments and cumulants of our random variable
since we're given its characteristic function.\end{rem}
\begin{thm}
Given $n\in\mathbb{N},N\in\mathbb{N},\mathbf{x}\in\mathbb{N}^{n}$,
then among all $\mathbf{x}$'s partitions that divide by $N>0$ with
remainder $j$, the correlation of the sign of $x_{1},x_{2}$ (wlog)
on those partition is given by:
\begin{equation}
-\intop_{0}^{\pi}t^{2}\sin\left(x_{1}t\right)\sin\left(x_{2}t\right)\prod_{k=3}^{n}\cos\left(x_{k}t\right)dt
\end{equation}
\end{thm}
\begin{proof}
Obtained immediately by differentiating (1) and differentiating wrt
$x_{1},x_{2}$ and integrating wrt $t$.
\end{proof}

\section{Estimating \#SAT}

The numbers produced by the reduction from \noun{\#sat} to \noun{\#part}
have digits that does not exceed 4, and if using radix 6, they never
even carry. Therefore the very same digits produced by the reduction
can be interpreted in any radix larger than 5, being reduced to a
different \noun{\#part }problem. Still, it is guaranteed that the
number of solution to those \noun{\#part} problems are independent
of the radix, as they're all reduced from the same \noun{\#sat} problem.
This property might be used to approximate \noun{\#sat} using results
as Theorem 3.3. We can obtain the number of partitions that their
size divides a given number $N$ in polynomial time wrt $n$ (the
number of numbers to partition) and the number of digits of $x_{k}$.
Nevertheless, it takes exponential time in the number of digits of
$N$. 

The probability that there exists a partition with nonzero size that
is divisible by a given prime $p$ is roughly
\begin{equation}
{\cal P}\left[p|\left\langle \mathbf{x},\sigma\right\rangle \right]\approx1-\left(1-\frac{1}{p}\right)^{2^{n}}
\end{equation}

taking $K$ reductions of a single \noun{\#sat} problem instance and
a set $P$ of primes, the probability that on reductions there exists
a partition with size divisible by a given prime $p$ that is not
a zero partition is therefore roughly
\begin{equation}
\prod_{p\in P}\prod_{k=1}^{K}{\cal P}\left[p|\left\langle \mathbf{x}_{k},\sigma\right\rangle \right]\approx\prod_{p\in P}\left[1-\left(1-\frac{1}{p}\right)^{2^{n}}\right]^{K}
\end{equation}
\[
\leq\exp\left(-K\sum_{p\in P}\left(1-\frac{1}{p}\right)^{2^{n}}\right)
\]
recalling that for $x\in\left[0,1\right]$ we have $e^{-x}\geq1-x$.
This doesn't seem to be helpful since it seem to require exponentially
many or exponentially large primes or reductions. However, if Conjecture
3.4 is true, then we can bound our hueristic approximation with rather
\begin{equation}
\exp\left(-K\sum_{p\in P}\left(1-\frac{1}{p}\right)^{\binom{n}{\frac{1}{2}n}}\right)
\end{equation}

\section{Discussion}

On Theorem 3.3 we have seen that we can efficiently query for the
number of partitions that divide by $N$ with remainder $j$. It is
interesting to see that positively solving \noun{part} (resp. \noun{sat})
by guesses is straight-forward: we just try partitions (resp. substitutions)
and if we're lucky to find a zero partition (resp. SAT) then we solved
the problem. On the other hand, how can we do one trial and possibly
decide that the set is unpartitionable (resp. UNSAT)? Our analysis
suggest such a method. If we query for the number partitions that
are divisible by $N$ with $j=0$ and get zero, then we know that
the set is unpartitionable. Similarly, if we do the same for $j\neq0$
and happen to get $2^{n}$, we know that $\mathbf{x}$ does not have
a zero partition. Those trials are argubaly independent due to the
pseudo-randomness of the mod operation.

On Section 4 \noun{\#part} (equivalently, \noun{\#sat}) is reduced
into a problem of computing the $k$'th digit of the result of the
multiplications of numbers of the form $100\dots001$, i.e. two ones
only and zeros between them. In fact, this result is independent on
the radix chosen (given it is not too small). On this setting, the
number of zeros has polynomial amount of digits wrt the problem's
input size, while the number of multiplicands is also polynomial.

On Section 5 we showed that \noun{p!=np} implies a result of nonexistence
of certain complex analytic functions. \noun{p!=np} also implies impossibility
to decide in polynomial time whether a definite integral (with bounded,
periodic, and polynomially evaluated integrand) equals either zero
or infinity, and proved that a single evaluation of $\psi$ can be
done in polynomial time.

On Section 6 we computed some asymptotic bounds that might be useful
in further analysis.

On Section 7 we practically showed how it is possible to express correlations
among different variables in a CNF formula, though we did not give
a full development of this idea.

On Section 8 we have seen that if Conjecture 3.4 is true, then we
can give heuristic having approximate exponentially convergent probabilistic
estimation to \noun{\#sat}, by taking advantage of the modular formulas
we derived in Theorem 3.3. Interestingly, this method reveals relatively
very little information about any single \noun{\#part} problem, since
we use reduce the \noun{\#sat} problem instance into many \noun{\#part}
instances, taking advantage on the reduction promising us \noun{\#part}
problems with quite different modular properties, yet with exactly
the same number of zero partitions.

For additional further research, by Theorem 3.3 we can get successive
estimates to the integral by selecting e.g. primes $N=2,3,5,\dots$.
We could then accelerate this sequence using Shanks, Romberg, Pade
or similar sequence-acceleration method.

\subsection*{Acknowledgments}

Thanks to Avishy Carmi and HunterMinerCrafter for many valuable discussions.

\appendix

\section{Appendix}

\subsection{Reductions}

\paragraph{Reduction of \noun{\#sat} to \noun{\#subset-sum}}

Given variables $x_{1},\dots,x_{l}$ and clauses $c_{1},\dots,c_{k}$
and let natural $b\geq6$. we construct a set $S$ and a target $t$
such that the resulted subset-sum problem requires finding a subset
of $S$ that sums to $t$. The number $t$ is $l$ ones followed by
$k$ 3s (i.e. of the form $1111\dots3333$). $S$ contains four groups
of numbers $y_{1},\dots,y_{l}$, $z_{1},\dots,z_{l}$, $g_{1},\dots,g_{k}$,
$h_{1},\dots,h_{k}$ where $g_{i}=h_{i}=b^{k-i}$, and $y_{i},z_{i}$
are $b^{k+l-i}$ plus $b^{m}$ for $y_{i}$ if variable $i$ appears
positively in clause $m$, or for $z_{i}$ if variable $i$ appears
negated in clause $m$. Then, every subset that sum to $t$ matches
to a satisfying assignment in the input CNF formula and vice versa,
as proved in \cite{key-1}.

\paragraph{Reduction of \noun{\#subset-sum} to \noun{\#part}}

Given $S,t$ as before and denote by $s=\sum_{x\in S}x$ the sum of
$S$ members, the matching \noun{part} problem is $S\cup\left\{ 2s-t,s+t\right\} $.
Here too all solutions to both problems are preserved by the reduction
and can be translated in both directions.

\subsection{Miscellaneous }
\begin{thm}
Let $Z^{\mathbf{x}}$ be the number of zero partitions of a vector
of naturals $X$. Let $D_{x}^{\mathbf{x}}$ be the number of zero
partitions of $X$ after multiplying one if its elements by two, where
this element is denoted by $x$. Let $A_{x}^{\mathbf{x}}$ be the
number of zero partitions of $X$ after appending it $x$ (so now
$x$ appears at least twice). Then 
\begin{equation}
Z^{\mathbf{x}}=D_{x}^{\mathbf{x}}+A_{x}^{\mathbf{x}}
\end{equation}
\end{thm}
\begin{proof}
Denote
\begin{equation}
\psi\left(x_{1},\dots,x_{n}\right)=2^{n}\intop_{0}^{\pi}\prod_{k=1}^{n}\cos\left(x_{k}t\right)dt
\end{equation}

then, using the identity $\cos2x=2\cos^{2}x-1$:
\begin{equation}
\psi\left(x_{1},\dots,2x_{m},\dots,x_{n}\right)=2^{n}\intop_{0}^{\pi}\cos\left(2x_{m}t\right)\prod_{k\neq m}^{n}\cos\left(x_{k}t\right)dt
\end{equation}
\[
=2^{n}\intop_{0}^{\pi}\left[2\cos^{2}\left(x_{m}t\right)-1\right]\prod_{k\neq m}^{n}\cos\left(x_{k}t\right)dt
\]
\begin{equation}
\implies\psi\left(x_{1},\dots,x_{m},\dots,x_{n}\right)-\psi\left(x_{1},\dots,2x_{m},\dots,x_{n}\right)=
\end{equation}
\[
2^{n+1}\intop_{0}^{\pi}\cos^{2}\left(x_{m}t\right)\prod_{k\neq m}^{n}\cos\left(x_{k}t\right)dt=\psi\left(x_{1},\dots,x_{m},\dots,x_{n},x_{m}\right)
\]
\end{proof}

\end{document}